\documentclass[amssymb,prb,preprint,aps,12pt]{article}
\usepackage{epsfig}
\usepackage{amsmath,amssymb,enumerate,mathrsfs
}
\usepackage {latexsym}
\usepackage{amssymb}
\usepackage{graphics}
\usepackage{amsfonts}
\usepackage{amsthm}
\newtheorem{theorem}{Theorem}[section]

\newtheorem{proposition}[theorem]{Proposition}

\begin{document}
\title{ Soliton  in a Well. Dynamics and Tunneling.}
\author{V. Fleurov \\
Raymond and Beverly Sackler Faculty of Exact Sciences, \\ School of
Physics and Astronomy, \\ Tel-Aviv University, Tel-Aviv 69978
Israel. \and A. Soffer\\
Department of Mathematics,\\ Rutgers University, New Brunswick, NJ
08903,USA} \maketitle
\begin{abstract}
We derive the leading order radiation through tunneling of an
oscillating soliton in a well. We use the hydrodynamic formulation
with a rigorous control of the errors for finite times.
\end{abstract}

\maketitle

\section{Introduction}

The dynamics of nonlinear dispersive systems pose a great challenge
to mathematics and physics. A large number of possible phenomena,
which include the formation of solitons or other coherent
structures, the blow up solutions, multichannel scattering, and
nonlinear tunneling, makes the development of a "general
theory" practically impossible. We then look at specific systems.

Here we will consider the motion of a soliton, described by the
Nonlinear Schr\"odinger (NLS) equation, moving inside a
one-dimensional well, which is however truncated so as to allow
tunneling.

This problem was considered in great detail by \cite{GS05,GS06,GS07}. They showed that in the semiclassical parameter regime a soliton, which is close to the bottom of the potential well, will stop moving as time goes to infinity. The process of stopping is due to radiation out of the well via tunneling. When viewed formally, the coupling of the bound state to the continuum shows up at $1/\omega$ order of perturbation expansion. Here $\omega = 2\pi/T$ where $T$ is the oscillation period. It was shown in \cite{GS05,GS06} (see also \cite{GW}) that the corresponding time dependent resonance theory \cite{SWGA,SW99} requires $1/\omega$ orders of iteration.

Our approach to this problem is different. We only solve the problem up to a finite time, though very long compared to the oscillation period $T$. In fact our results apply to times of the order $t_0 = 1/(\omega\varepsilon)$ where $\varepsilon $ is the oscillation amplitude. Previous work on the finite time semiclassical soliton was developed in detail and for a general potential in \cite{FGJS}. In that paper only the classical trajectory for the soliton is derived, but not the dissipation due to quantum tunneling. We prove using the hydrodynamic formalism of the quantum tunneling \cite{FS05,DFSS07,DFSF09,DFFS10} that the radiation outside the well is outgoing and find its leading order behavior. The fact that the radiation is purely outgoing is then used to prove that it implies the monotonic decay in time of the internal energy of the soliton and of its oscillation amplitude. These estimates hold up to times of order $t_0$. So the extension to an arbitrary time is possible, since it implies that in the time interval $0 \leq t \leq t_0$ the amplitude $\varepsilon(t)$ goes monotonically down.

The control of the errors in the hydrodynamic formulation is done here for the first time. The key difficulty is to estimate the error due to the quantum potential. The latter is defined as
$$
Q(|\psi|) = - \frac{\partial_x^2 |\psi|}{|\psi|}
$$
for the wave function $\psi$. Let $\psi$ be written as
$$
\psi = S(x - a(t))(1 + \chi(x,t))e^{\imath \theta}
$$
where $S(x)$ is the static soliton profile at the bottom of the trap. The control of the error in $Q(|\psi|)$ is reduced to estimating the multiplicative type perturbation $\chi$.

A further consequence of our analysis is that with a particular choice of the potential at the transition region, we observe a significant suppression of tunneling, which stabilizes the soliton oscillation dynamics for a long time.

\section{Soliton in the closed potential well $V(x)$}

We consider in this section motion of a soliton in a closed well
from which tunneling is impossible. The dynamics in this case is
governed by the Gross-Pitaevskii (GP) equation (Nonlinear
Schr\"odinger equation - NLS)
\begin{equation}\label{NLS}
i \partial_t\psi = - \partial_x^2\psi + \lambda
|\psi|^2\psi + V(x) \psi(x,t)
\end{equation}
(the Planck constant $\hbar =1$ and the mass $m = 1/2$). Here we
assume that we know the real function $S(x)$ of the soliton sitting
at the minimum at $x=0$ of the trap potential $V(x)$, which solves
the equation
\begin{equation}\label{NLS-a}
-E S(x) = - \partial_x^2 S(x) + \lambda S^3(x) + V(x) S(x)
\end{equation}

Then the approximate solution corresponding to a moving soliton is
looked for in the form
\begin{equation}\label{solution1}
\psi(x,t) = S(x - X(t)) e^{+ i(E t + \Phi)}
\end{equation}
with a real phase $\Phi$. $X(t)$ defines the position of the soliton
at a given time $t$. Substituting this solution into Eq. (\ref{NLS})
and separating the real and imaginary terms we get two equations

\begin{equation}\label{HJ-1}
\left[\partial_t\Phi +  \left( \partial_x \Phi
\right)^2  + W(x,t)\right] S(x - X(t)) = 0
\end{equation}
where $W(x,t) = V(x) - V(x - X(t))$ and
\begin{equation}\label{HJ-2}
\left[\partial_t X(t) - 2\partial_x\Phi \right] \partial_x S(x -
X(t)) =  \partial_x^2 \Phi S(x - X(t))
\end{equation}

Assuming now that the soliton $S(x - X(t))$ is very narrow even
function centered around $X(t)$ we may integrate Eq. (\ref{HJ-1})
over $x$ and approximately obtain the equation
\begin{equation}\label{HJ-3}
\partial_t\Phi(X,t) + \left( \partial_x \Phi(X,t)
\right)^2  + V(X) = 0
\end{equation}
which is a Hamilton-Jacoby equation describing the motion of a
particle with the mass $1/2$ at the center $X$ of the soliton. $W(X)
= V(X) - V(0)$ should have been written in Eq. (\ref{HJ-3}), however
the constant $V(0)$ in the potential may be always omitted. The
phase $\Phi(X)$ plays the role of mechanical action and the velocity
of the soliton motion is
$$
\dot X = 2\partial_X \Phi(X).
$$
Hence for $x$ close to $X$ the Ansatz
\begin{equation}\label{phase1}
\Phi(x,t) \approx  \dot X(t) x/2 + F(t)
\end{equation}
may be applied. Here the function $F(t)$ depends only on $t$. As a
result, Eq. (\ref{HJ-2}) becomes an identity. Substituting Eq.
(\ref{phase1}) into (\ref{HJ-1}) we get in the same approximation
that
\begin{equation}\label{phase2}
F(t) =\frac{1}{ 2} \left. \dot X X\right|_0^t - \int\limits_0^t
\left[\frac{1}{4}\dot X(s)^2 - V(X(s)) \right] ds
\end{equation}
which really has the form of a mechanical action with a boundary
term which does not affect equation of motion for $X(t)$.

Correspondingly the total energy of the soliton within the trap
becomes
$$
E_{tot} = -E - V(0) + \frac{1}{4} \dot X(t)^2 + V(X(t)) = \mbox{const}
$$

The corrections to the above approximate derivation are due to the
finite width of the soliton,
\begin{equation}\label{corrections}
\widetilde\beta^{-2} = \int dx x^2 S(x).
\end{equation}
which is supposed to be small compared to the characteristic scale
of the potential.

It is important to emphasize that in the case of harmonic
potential $V(x) = \omega^2 x^2/4$ the approximate solution
(\ref{solution1}) and (\ref{phase1}) for a soliton moving in the
trap becomes exact with
$$
X(t) = C_0\cos 2\omega t
$$
where $C_0$ is the oscillation amplitude and
$$
F(t) = - \frac{1}{8}\omega C_0^2 \sin 2\omega t.
$$
This can be verified by a direct substitution into Eqs. (\ref{HJ-1})
and (\ref{HJ-2}). As for the total energy it is obviously
$$
E_{tot} = -E - V(0) + \frac{1}{4}C_0^2\omega^2.
$$
For other types of nontrivial exact solutions, see e.g. \cite{r05,K09}.

\section{Estimates of the corrections}

\subsection{"Time dependent case"}

In order to estimate the precision of the applied procedure we
return to Eq. (\ref{NLS}) which describes a soliton oscillating
within a trap and radiation due to tunneling. Hence the chemical
potential $E(t)$, (Energy), and the shape of function $S(x)$ vary with time
$t$.

The general solution will be looked for using the following Ansatz (see \cite{GS05,GS06,GS07,BS03,SW04} for related constructions):
\begin{equation}\label{solution2}
\psi = S_E(x-X(t),t)(1 + \chi) e^{i\theta},
\end{equation}
where
$$
\psi_{t=0} = S(x - X_m,0).
$$
with some initial value $X_m$. $S_E(x,t)$ is the solution of the nonlinear soliton equation
$$
 -E(t)S(x,t) = (- \partial_x^2 + V(x) + F(S))S.
$$
for a given value of $E(t)$, which may vary adiabatically slow with
time. We distinguish here the rapid time dependence due to
oscillations of the soliton within the trap described by the shift
$X(t)$ and slow variation of the shape of the function $\psi_{t=0}$
due to exponentially weak tunneling.

Substituting (\ref{solution2}) into (\ref{NLS}) (the subscript $E$
is suppressed for the sake of brevity) we get
\begin{multline}
 i \dot{S} (1+\chi) e^{i\theta} +i S \dot{\chi}
e^{i\theta} -S(1+\chi)
\dot{\theta} e^{i\theta}=\\
\left[ -\partial_x^2 \psi + V\psi + [F(S(1+\chi))-F(S)+F(S)]\psi
\right]\\
= e^{i\theta}\bigl[  \left(- S'' + VS + F(S)S\right)(1+\chi)\\+
S\left(- \chi'' - 2\frac{S^\prime  }{S } \chi'  + \lambda S^2(\chi
+\chi^*)\right)\\
+ \lambda S^3(\chi^2 +2|\chi|^2+|\chi|^2\chi)\bigr]\\
+ e^{i\theta} S(1+\chi)\left[(\theta^{\prime})^2 -i\theta^{''} -2i
\theta^{\prime} \frac{ \widetilde{S}^{\prime}}{\widetilde{S}}
\right]
\end{multline}
where $\widetilde{S} = S(1 + \chi)$. Here ${\cdot }$ stands for the time derivative and $'$ stands for the $x$ derivative.

Now, we use the Ansatz
$$
\theta(x,t)=\int_0^t E(s) ds + \gamma(t) + \varphi (x,t),
$$
such that
$$
2\partial_x \theta = v_{\theta}=v_{\varphi} =2 \partial_x \varphi.
$$
Then, the Schr\"odinger equation becomes
\begin{eqnarray}\label{NLS-chi}
\imath \dot S (1+\chi)& + &\imath S \dot \chi -S(1+\chi) (\dot{
\varphi}+\dot {\gamma}) = (V-V_t)(1 + \chi)S \nonumber
\\
&+& S \left( - \chi'' -2(S^{\prime}/S) \chi'
+\lambda S^2(\chi +\chi^*)\right) \nonumber
\\
&+& \lambda S^3(\chi^2 +2|\chi|^2+|\chi|^2\chi) \nonumber
\\
&+& S(1+\chi)v_{\theta}^2/4 - \imath (\widetilde{S}^{\prime}
v_{\theta} + \widetilde{S} v_{\theta}^{\prime}/2).
\end{eqnarray}

The complexified equation for the vector
$$
{\bf K} =
\begin{pmatrix}
\chi \\ \chi^*
\end{pmatrix}
$$
can then be written as
\begin{equation}\label{modulation}
\imath  S^{-1} \dot{\textbf{S}} +\imath \dot{\bf
K} = H{\bf K} + {\cal O}(b{\bf K},b) + \mathcal{NL}({\bf
K}) +{\cal O} \left(\frac{S'}{S}v_\theta + v'_{\theta}\right).
\end{equation}
where
$$
H:=
\begin{pmatrix}
- \partial_x^2 +\lambda S^2-2(S^\prime/S)\partial_x & \lambda S^2\\
-\lambda S^2 & \partial_x^2 -\lambda S^2 +2(S^\prime/S)\partial_x
\end{pmatrix},
$$
$\textbf{S}$ stands for the vector with components $S,-S$; $b$
stands for time dependent, space localized functions, which (are
expected to) have a zero limit at infinite time. $ \mathcal{NL}$ stands for
nonlinear terms in ${\bf K}$.

\medskip

\paragraph{ Orthogonality Conditions}

\medskip

Let $H_D:=S(H - E \sigma_z )S^{-1}$
then
\begin{align}
H_D \begin{pmatrix} S\\ -S \end{pmatrix}&=0, \\
H_D^2 \begin{pmatrix} \partial_E S \\ \partial_E S \end{pmatrix}&=0.
\end{align}
Similar identities hold for the adjoint operator $H_D^*.$ The
corresponding eigenfunctions are obtained from the previous ones, by
applying the Pauli matrix $\sigma_z$.

$H_D$ has also two generalized eigenfunctions, which are small
perturbations of the vectors
$$
{\bf \eta}_1=\begin{pmatrix} \partial_x S \\ \partial_x S
\end{pmatrix}, \qquad {\bf \eta}_2= \begin{pmatrix} xS\\ -xS \end{pmatrix},
$$
with eigenvalues close to zero, provided the soliton is sufficiently
narrow. See \cite{GS05,GS06}.
$$
{\bf \xi}_1=\begin{pmatrix} \partial_E S \\ \partial_E S
\end{pmatrix}.
$$

We now impose the main {\em orthogonality conditions}, which make
the perturbation $\chi$, orthogonal to the soliton $S$.
$$
\mbox{Re} \langle S^2,\chi \rangle =0,
$$

$$
\mbox{Im} \langle \chi,SS^{'}_E\rangle=0,
$$

$$
\mbox{Im} \langle S_x',S\chi\rangle=0.
$$
where $\langle\varphi, \psi \rangle$ denotes the usual $L^2$ scalar
product of $\varphi$ and $\psi$.
The first two conditions determine the equations of $\dot E, \dot \gamma.$
The other conditions determine the boost and translation of the Soliton at time t.

For example, multiplying (\ref{modulation}) by $S$, taking the scalar product
with the vector $\sigma_3\xi_1,$ and using the orthogonality
condition, using also the equation
\begin{equation}\label{S dot}
\partial_t S(x-X(t),E(t)) = - \partial_x S(x - X(t)) \dot X(t)+
\partial_E S \dot E.
\end{equation}
for $\dot S$, integration gives
\begin{equation}\label{Gamma dot}
2 \dot \gamma  \int  S \partial_E S dx
 = \langle \sigma_3\xi_1, S
\mathcal{O}(b{\bf K},b) + S \mathcal{NL}({\bf K})+
\mathcal{O}(Sv_{\theta})\rangle.
\end{equation}

Going back to equation (\ref{NLS-chi}), if we choose
$$
\theta_h(t)=\int_0^tE(s)ds + \frac{1}{2}x\dot X(t) + q(x,t),
$$
and$$
\psi=S(1+\chi)e^{\imath \theta_h}=S|1+\chi|e^{\imath q+\imath\theta_h}=S|1+\chi|e^{\imath \phi},
 $$
we derive
$$
v_{\theta_h}=2\partial_x \theta_h=\dot X(t)=v_\phi-v_q.
$$
Finally, we derive the following equations for $\chi$:
\begin{equation}\label{chi-eq1}
0=
$$$$
-\partial_x \left[[V(x)-V(x-X)-x\ddot X/2] \phantom{\frac{S}{\tilde S}}\right.+
$$$$
\left.[\frac{S}{\tilde S}(\mbox{Im} \dot\chi+H_{\chi} \mbox{Re} \chi + 2\lambda S^2 \mbox{Re} \chi
+\mathcal O(\chi^2))]+\tilde S^{-1} \partial_E S \dot E \mbox{Im} \chi \right].
\end{equation}
From this we derive
\begin{equation}\label{chi-eq2}
-\partial_t \mbox{Im} \chi = H_{\chi} \mbox{Re} \chi+(2\lambda S^2+\dot \gamma) \mbox{Re} \chi
+\mathcal O(\chi^2)+
$$$$
\tilde S^{-1} \partial_E S \dot E \mbox{Im} \chi+[V(x)-V(x-X)-x\ddot X/2+\dot\gamma].
\end{equation}

\section{Hydrodynamic representation}

In order to find the tunneling flux from an open potential
$V_{ext}(x)$ with a soliton oscillating within the well we apply the
procedure similar to that used in our previous papers\cite{FS05,DFSS07,DFSF09,DFFS10}. First the GP equation (\ref{NLS}) for a complex function
$\psi= |\psi| \exp\{i\phi \}$ can be represented as two
hydrodynamic equations,\cite{m27,m75},
\begin{equation}\label{HR1}
\partial_t\rho(x,t) + \partial_x [\rho(x,t)v(x,t)] = 0.
\end{equation}
and
\begin{equation}\label{HR2}
\frac{1}{2}\partial_t v(x,t) + \frac{1}{2} v(x,t) \partial_x v(x,t) = - \partial_x
\left [V_{ext}(x) - \frac{\partial_x^2 \sqrt{\rho(x,t)}}{
\sqrt{\rho(x,t)}} + \lambda \rho(x,t)\right]
\end{equation}
for two real functions: the density $\rho(x,t) = |\psi|^2$ and the
velocity $v =v_{\phi}= 2 \partial_x \phi$.

\begin{theorem} { Hydrodynamic Equations.}
Let $\psi$ satisfy the NLS equation, with initial data of finite energy;
assume, moreover that the solution exists and is bounded.
Then, if for all $(x,t) \in [I_X,I_T],$ $\psi\ne 0$, then the hydrodynamic equations (\ref{HR1}) and (\ref{HR2}) for
$\rho,v$ have a unique solution identified with the corresponding NLS, via the relations
$ \rho=|\psi|^2, \psi=|\psi|e^{\imath \phi}, v=2\partial_x\phi.$
\end{theorem}

For a proof see \cite{Sof}. The general problem of approximating the Schr\"odinger type equations with kinetic type or fluid equations,
has been studied extensively. It is important in understanding semiclassical methods\cite{GJ,Gols}, large particle system limits
and more\cite{Per,BGP,DLev}.

In order to solve the problem we choose $|\psi| = S_0(x - X(t))$
obtained near the bottom of the potential Fig. 1. That is $S_0(x)$
is an exact solution stationary at the bottom of the potential well.
We assume that the initial condition is that the solution at time
$t=0$ is given by the soliton profile $S(x-X_m)$ with $X_m$ small.
Here, it is assumed that the potential have a local minimum at zero.
(The function $S_0$ is taken to solve the nonlinear soliton equation
(\ref{NLS-a}).) Then
\begin{equation}\label{Euler}
\partial_t v(x,t) +  v(x,t) \partial_x v(x,t) = -2 \partial_x (U_1+U_2)
\end{equation}
with
\begin{equation}
U_1 = \left[V(x) - V(x-X(t)) \right]
\end{equation}
and
\begin{equation}
U_2= {\cal F}(S(1 + \chi))- {\cal F}(S) + Q  + V(x - X(t)) + {\cal
F}(S)
\end{equation}
where $Q = - \frac{\partial_x^2 |\psi|}{|\psi|}$ is the quantum
potential (QP) and ${\cal F}$ stands for the nonlinear term. We also
use the ansatz
\begin{equation}\label{CHI}
\psi= S(x - X(t))(1 + \chi(x,t)) e^{i\theta}
\end{equation}
where $\theta = \theta_h + q$, with $\theta_h$ solving the harmonic
problem with the frequency $\omega$.

The derivative
$$
\partial_x[V_h(x)-V_h(x-X(t))] = \omega^2 X(t)/2
$$
for some constant $\omega$, so that $(m=1/2)$,
$$
\partial_t^2 X(t ) = - \omega^2 X(t)
$$
and
$$
v(x,t)= \dot X(t) + w(x,t).
$$
We use that $Q(S)=-\partial_x^2 S/S,$ and
$$
- \frac{\partial_x^2 S}{S}  + V(x - X(t)) + {\cal
F}(S) = -E = const
$$
and $\psi = S(1 + \chi) e^{i \theta}$, as well as the estimate of
$Q(\psi) - Q(S)$ to be presented in the next subsection. Then the
equation for $w$ becomes
\begin{equation}\label{6}
\partial_t w(x,t) + \dot X(t) \partial_x w(x,t) + w(x,t) \partial_x w(x,t) = -2
\partial_x W_1(x,t)
\end{equation}
with
$$
W_1(x,t)= V(x)-V_h(x) + V_h(x-X(t)) - V(x-X(t)) +
$$$$
{\cal F}(S(1+\chi))- {\cal F}(S) + Q(\psi)-Q(S),
$$
where the correction $w(x,t)$ appears due to tunneling and
corresponds to the escape of the matter from the trap.
Alternatively, the equation for $v$ is
\begin{equation}\label{7}
\partial_t v(x,t) +  v(x,t)\partial_x v(x,t) =
$$$$
-2 \partial_x [W(x,t)(1 +\chi )+{\cal O}(S\chi)+Q(\psi)-Q(S) + {\cal O}(\chi^2 )].
\end{equation}
with
$$
W(x,t)= V(x)-V(x-X(t))
$$
The solution of this equation can be found from the solution of the
corresponding classical mechanical system. After approximating the
quantum potential as $Q(\psi)\sim Q(S)+ {\cal O}(\chi)$, we obtain
the equation
\begin{equation}\label{Burgers}
\partial_t v(x,t) +  v(x,t)\partial_x v(x,t) =
$$$$
-2 \partial_x [W(x,t)(1 + {\cal \chi}) - L( \mbox{Re} \chi) +{\cal O}(S\chi^2+ \chi \partial_x \chi)].
\end{equation}
for the velocity $v$.
$$
L(\mbox{Re} \chi)\equiv H_{\chi} \mbox{Re} \chi +2\lambda S^2 \mbox{Re} \chi +{\cal O}(\dot E S \mbox{Re} \chi).
$$
See Section 4 for the estimate of $Q(\psi)-Q(S).$

To find the solution for the classical system,
let us assume first that the potential looks like piecewise
harmonic: It is harmonic well, around the origin, up to distance of
order $1/\omega$, then a transition region, inverted harmonic, on an
interval of order 1, and then goes down as the same harmonic
potential, all the way to zero.
$$
V(x) = \left\{
\begin{array}{cc}
\frac{1}{4}\omega^2 x^2, &\ |x| \leq \frac{1}{\omega}, \\
- \frac{1}{4} \omega^2 (x \mp \frac{1}{\omega})^2 \pm \omega^2 {\cal
O}(x^3) & |x \mp \frac{1}{\omega}|  \leq \frac{\delta}{\omega},
\\
\frac{1}{4}\omega^2 (x - \frac{2}{\omega})^2, &\ \frac{2}{\omega}
\geq x > \frac{1
+ \delta}{\omega}, \\
\frac{1}{4}\omega^2 (x + \frac{2}{\omega})^2, &\ - \frac{2}{\omega}
\leq x < - \frac{1 + \delta}{\omega}, \\
0 & |x| > \frac{2}{\omega}
\end{array}\right.
$$
with $0 < \delta \ll 1$.
\begin{theorem}{Burgers' Equation.}

Consider the Burgers equation defined in (\ref{Burgers}), and with zero initial data.
The solution of this equation, without the second order terms corrections,
in the transition region $\frac{2}{\omega}\geq x > \frac{1 + \delta}{\omega}$,
is given by
\begin{equation}\label{tunnel-vel}
v(x,t)=\frac{1}{2}\alpha\epsilon^2\omega^2 \left[ t-(t^2-\beta\epsilon^{-2}\omega^{-2}(x-(1+\delta)/\omega))^{1/2} \right]
\end{equation}
for $t^2\gg \beta\epsilon^{-2}\omega^{-2}(x-(1+\delta)/\omega))$
.
Moreover, for $x$ in the transition region, and $(x,t)$ not satisfying the above inequality, the velocity $v(x,t)$ is zero.
\end{theorem}

\textbf{Remark}
The above result shows that tunneling is suppressed for this choice of the potential barrier, for times of order $\beta^{1/2}\epsilon^{-1}\omega^{-1/2}.$

\begin{proof}

In this case, the potential field acting on
the classical particle in the transition region is approximately
given by
\begin{equation}
V(x) - V(x-X)
$$$$
 = V'(x)X - (1/2) V''(x) X^2+{\cal O}(X^3)
 $$$$
 \sim \omega^2 x X - \alpha(x)
\omega^2 X^2 + C(t) + {\cal O}(X^3)
\end{equation}
with a positive $ s_n(x)\alpha$, since in the transition region,
the third derivative of $V$ is positive, of order $\omega^2$ on an
interval of a size ${\cal O}(1)$.
Here, $s_n(x)$ denotes the sign of $x.$
So, the force is outgoing,
\begin{equation}
F \sim - \omega^2 X + \alpha_0 \omega^2 X^2 = - \omega^2 X + [(1/2)
\alpha - \cos(2\omega t)]\epsilon^2 \omega^2.
\end{equation}
where we choose $\alpha(x)\sim\alpha_0 x+\beta_0,$ and
since $ X(t)\equiv \epsilon \sin \omega t$ with a small oscillation
amplitude $\epsilon$.

So, the classical particle goes out of the transition region, with
the velocity $v_0 \sim \epsilon ^2 (a/2)t_0+d\dot X$, the exit time
is $t_0$.  Hence
\begin{equation}
\frac{1}{2} \epsilon^2 \omega^2 t_0^2 \sim 1
\end{equation}
so that
\begin{equation}
t_0 \sim \frac{1}{\omega \epsilon}
\end{equation}
(The third derivative is of the order $\omega^2$ only on an interval
of size ${\cal O}(1))$. To the leading order, $d=1.$ So, outside the
transition region, and AFTER time $t'$, the acceleration is zero
and the velocity is therefore given by $v(x_0,t')$, and up to $\dot X$ is positive,
of the order $\epsilon^2 \omega^2 t',$  for $(x - 1/\omega) > 0$ .
So, to find the solution for the velocity at $(x,t)$, we get
$$
x-x_0=v(x_0,t')(t-t')
$$
for some $t'\le t.$
Since the acceleration is constant, when it is nonzero, by our choice of the potential,
we get that:
$$
x-x_0=\frac{1}{2}\alpha\epsilon^2\omega^2 t'(t-t').
$$
Solving for $t'$, gives TWO possible solutions. The solution of the Burgers equation is the one which minimizes the action.
The action is proportional to $(t')^3,$ and therefore is minimized by the smaller choice of $t'.$
This gives the result of the theorem.
\end{proof}

\subsection{Corrections to the quantum potential}

Here we show that assuming small $\chi$, linear corrections to the quantum
potential
$$
Q = - \frac{\partial_x^2 |\psi|}{|\psi|}
$$
make important contribution.
We will prove that:
\begin{theorem}
Suppose that $|\chi|<<1$, and the solution of the NLS is given as before by our Ansatz , in terms of the Soliton $S,$, and the correction
$\chi.$
Then the Quantum potential $Q$ is given by
\begin{equation}\label{QP3}
Q = - \frac{\partial_x^2 S}{S} +H_{\chi} \mbox{Re} \chi + {\cal
O}(\chi
\partial_x\chi) + \mbox{Re} {\cal NL}(\chi,\chi^*)
$$$$
Q= - \frac{\partial_x^2 S}{S} -\partial_t \mbox{Im } \chi -[V(x)-V(x-X(t))]
$$$$
+L(\mbox{Re}\chi)+ \mbox{Re} \mathcal{NL}(\chi,\chi^*)+ {\cal
O}(\chi\partial_x\chi)+{\cal O}(\mbox{Im} \chi \partial_t \mbox{Re}\chi).
\end{equation}
with
$$
H_\chi = - \partial_x^2 - \frac{2 \partial_x S}{S} \partial_x.
$$

\end{theorem}

\begin{proof}
 Using
$$
|\psi| = S |1 + \chi|
$$
we write
\begin{equation}\label{QP1}
Q = - \frac{\partial_x^2 S}{S} - \frac{2\partial_x S\partial_x |1 +
\chi|}{S |1 + \chi|} - \frac{\partial^2_x |1 + \chi|}{|1 + \chi|}
\end{equation}
We may now calculate the derivatives
$$
\partial_x |1 + \chi| = \frac{\partial_x \mbox{Re}\chi +
(\partial_x \mbox{Re} \chi) \mbox{Re} \chi + (\partial_x \mbox{Im}
\chi) \mbox{Im} \chi}{|1 + \chi|}
$$
and
$$
\partial_x^2 |1 + \chi| = \frac{\partial_x [\partial_x \mbox{Re} \chi +
(\partial_x \mbox{Re} \chi) \mbox{Re} \chi + (\partial_x \mbox{Im}
\chi) \mbox{Im} \chi]}{|1 + \chi|} - \frac{(\partial_x \mbox{Re}
\chi)^2}{|1 + \chi|^2} + O(\chi^3)
$$

\begin{equation}\label{der7}
Q = - \frac{\partial_x^2 S}{S} - \frac{2\partial_x S}{S} \frac{1}{|1 + \chi|} \left[
\partial_x \mbox{Re}\chi (1 + \mbox{Re}\chi) + (\partial_x \mbox{Im}\chi) \mbox{Im}\chi)
\right] -
$$$$
\frac{1}{|1 + \chi|} \left[\partial_x^2 \mbox{Re}\chi (1 +
\mbox{Re}\chi) + (\partial_x \mbox{Re}\chi)^2
+(\partial_x\mbox{Im}\chi)^2+ \mbox{Im}\chi\partial_x^2 \mbox{Im}\chi\right]
$$$$
+ {\cal O}(\chi^3)
\end{equation}
since
$$
(\partial_x\mbox{Re}\chi)^2 - \frac{(\partial_x\mbox{Re}\chi)^2}{|1
+ \chi|} = {\cal O}(\chi^3).
$$

Then the force coming from $Q$ is
$$
- \partial_x Q = - \partial_x \left( - \frac{\partial_x^2S}{S} - \frac{2\partial_xS}{S}
(\partial_x \mbox{Re}\chi) - (\partial^2_x \mbox{Re}\chi) \right)
$$$$
- \partial_x \left[\left(\frac{1}{|1 + \chi|} - 1 \right)\left(-
\frac{2\partial_x S}{S} (\partial_x \mbox{Re}\chi) - (\partial^2_x \mbox{Re}\chi)
\right)\right]
$$$$
- \frac{1}{|1 + \chi|}\left\{- \partial_x \left[ \mbox{Re}\chi
\partial_x^2 \mbox{Re}\chi + (\partial_x
\mbox{Im}\chi)^2\right]\right\} + {\cal O}(\chi^2 \partial_x\chi)+\partial_x{\cal O}(\partial_x^2(\mbox{Im}\chi)^2)
$$$$
= {\cal O}(\chi \partial_x\chi) + {\cal O}(\partial_t(\chi
\partial_x\chi))
$$

Next we eliminate $\partial_x^2 \mbox{Re}\chi$, $\partial_x^2 \mbox{Im}\chi$
in the above equations, using the equations
\begin{equation}\label{Chi eq Im}
-\partial_t \mbox{Im}\chi = H_{\chi} \mbox{Re}\chi +{\cal O }(S^2 \chi + \chi^2)+G(W(x,X(t))).
\end{equation}
\begin{equation}\label{Chi eq Re}
\partial_t \mbox{Re} \chi= H_{\chi} \mbox{Im}\chi +{\cal O }(S^2 \chi + \chi^2).
\end{equation}

Therefore the quantum potential becomes in the leading order
\begin{equation}\label{QP2}
Q = - \frac{\partial_x^2 S}{S} - \frac{2\partial_x S }{S} \partial_x
\mbox{Re} \chi - \partial^2_x \mbox{Re} \chi + {\cal O}(\chi
\partial_x \chi)
$$$$
+{\cal O}((\partial_x \mbox{Im}\chi)^2+ \mbox{Im} \chi\partial_t \mbox{Re}\chi).
\end{equation}

Now we use equation
$$
i \partial_t \chi = H_\chi \chi + {\cal L}(\chi) + {\cal NL}(\chi,\chi^*)+G(W(x,X(t)))
$$
where ${\cal L}$ and ${\cal NL}$ stand for linear and nonlinear terms in $\chi$.
It allows us to write
$$
-\partial_t \mbox{Im} \chi = -\left(\partial_x^2 + \frac{2 \partial_x
S}{S} \partial_x \right) \mbox{Re} \chi + \mbox{Re}(L(S\chi) +
\mathcal{NL}(\chi,\chi^*))+G(W(x,X(t))),
$$
which implies the statement of the Theorem.
\end{proof}

In order to proceed further on we will need some relations between
the correction $q$ to the soliton phase due to $\chi$ and the phase
$\eta$ of the latter. They are connected by equation
\begin{equation}\label{phase3}
1 + |\chi| e^{i\eta} = |1 + \chi| e^{iq}.
\end{equation}
We will also need the derivative
\begin{equation}\label{phase4}
\partial_x |1 + \chi| = \frac{1}{2\sqrt{1 + |\chi^2| + 2 |\chi| \cos
\eta}} [\partial_x |\chi|^2 + 2 \cos\eta \partial_x |\chi| + 2
|\chi| \sin \eta \partial_x \eta].
\end{equation}

Equation (\ref{phase3}) for complex functions is equivalent to
\begin{eqnarray}
1 + |\chi| \cos\eta = |1 + \chi| \cos q \label{phase5a},\\
|\chi| \sin\eta = |1 + \chi| \sin q \label{phase5b}.
\end{eqnarray}

\section{Hydrodynamic Modulation Equations}

In this section we combine the estimate for the velocity being
outgoing in the outer well, together with the equations for $\chi,$
$\dot E,$ $R,$ to prove that the soliton stabilizes. Recall the
ansatz for the solution, in terms of the soliton $S$ and the phase
$\theta$:

$$
\psi=S(1+\chi)e^{\imath \theta} :=S|1 + \chi|e^{\imath [\theta+q]}.
$$
Since we have shown that the velocity is outgoing, modulo the
oscillating part, we have that (choosing $\theta=\theta_h$),
$$
\mathrm{sign}(x)\partial_x q \sim \omega^2 \epsilon^2t >0, \ \ t
\leq t_0 \sim \frac{1}{\omega\epsilon} .
$$
\begin{theorem}{Soliton Energy Decay.}

Let the solution of NLS be such that $|\chi|+|\chi^{'}|<<1, \partial_x q>0.$
Assume, moreover that the soliton solution $S$ is stable.
Then $$
\dot E\leq -c(\omega) \mbox{Re} \chi\partial_x q(x_0)
$$
for some $x_0$ at the transition region, and $c(\omega)$ is a constant of order
$\mathcal O(e^{-bx_0})$, and $b$ is a positive constant related to the decay rate of the soliton (in $x$) at large distances.
\end{theorem}
\begin{proof}

We have
\begin{equation}\label{mass1}
\partial_t \int_{-\infty}^x |\psi|^2 dx = - j(\psi(x)) = -2 S^2 |1 +
\chi|^2(\partial_x \theta + \partial_x q)
$$$$
=-2S^2(\partial_x \theta + \partial_x q)-2S^2(\mbox{Re} \chi)\dot X(t)-S^2|\chi|^2(\dot X(t)+v_q)-4S^2(\mbox{Re} \chi)\partial_x q.
 \end{equation}
Using the equation
\begin{equation}\label{psi}
|\psi|^2=S^2+2S^2 \mbox{Re} \chi + S^2|\chi|^2,
\end{equation}
and combining equations (\ref{S dot}), ( \ref{mass1}), (\ref{psi})
and $2\partial_x \theta_h = \dot X$ we derive that
\begin{equation}\label{mass2}
\partial_t \int_{-\infty}^x|R|^2 + 2\dot E(t) \int_{-\infty}^x S
\partial_E S =-S^2[|1+\chi|^2-1]\dot X(t)
$$$$
- \left[ 2S^2 \partial_x q + 2S^2(2 \mbox{Re} \chi + |\chi|^2)
 \partial_x q + \partial_t \int_{-\infty}^x
2S^2 \mbox{Re} \chi \right].
\end{equation}
As a result of the orthogonality conditions, the modulation equation for $R=S\chi$ reads
\begin{equation}\label{R norm}
\partial_t \int_{-\infty}^x |R|^2 dx = - j(R) + \mathcal{O}(\mathcal{NL}(R)) +
\mbox{Im} \left\{{R^* S[F(|\psi|)-F(S)]}\right\}
$$$$
=\partial_t\int_{-\infty}^{+\infty}|R|^2 dx-\partial_t\int_x^{\infty}|R|^2 dx=-\partial_t\int_{-\infty}^{+\infty}|S|^2 dx-\partial_t\int_{x}^{+\infty}|R|^2 dx.
\end{equation}
with
\begin{equation}\label{J}
j(R) = S^2 j(\chi)=- \imath (R^* \partial_x R - R \partial_x R^*).
\end{equation}

Next, we use that
$$
1 + \chi = S^{-1}| \psi| e^{\imath q} = |1+\chi| e^{\imath q}.
$$
So that $ \mbox{Re} \chi = |\chi| \cos \eta= |1+\chi|\cos q -1$,
$\mbox{Im} \chi = |\chi| \sin \eta =|1+\chi| \sin q$. Using the
assumption that $\chi$, $\partial_x \chi$, and $q$, are all small in
the relevant region of space and time, we arrive at
\begin{equation}\label{Im chi}
\partial_x \mbox{Im} \chi  = (\partial_x q)|1+\chi| \cos q +
(\partial_x |1 + \chi|) \sin q =
$$$$
\partial_x q\left[ |1+\chi|\cos q\right]+2\partial_x \mbox{Re} \chi |\chi|\sin\eta +{\cal O}(\chi^2)
$$$$
=\partial_x q +|\chi|\cos \eta \partial_x q+2|\chi|\partial_x q \sin^2\eta(2\cos\eta+\sin\eta)+{\cal O}(\chi^2)
$$$$
=
\partial_x q (1 + {\cal O}(\mbox{Re}\chi)) +
{\cal O}(q^2) \partial_x q + {\cal O}(q)  (\partial_x \mbox{Re}\chi)
(1 + {\cal O}(\chi)).
\end{equation}

We now compute the derivative of $\mbox{Re} \chi$ with respect to
time, to leading order. For this we use the equality
\begin{equation}\label{der1}
-2\partial_t \int_{-\infty}^x S^2 \mbox{Re} \chi dx = 2 \partial_t
\int_x^{\infty} S^2 \mbox{Re} \chi dx,
\end{equation}
which follows from the orthogonality $(S^2,\mbox{Re}\chi).$ It is obvious that the
integral in the right hand side of this equality is exponentially
small. Then
\begin{eqnarray}\label{der2}
+ 2 \partial_t \int_x^{\infty} S^2 \mbox{Re} \chi dx & = & \nonumber \\
& 2 &\int_x^{\infty} S^2 \langle [(- \partial_x^2 - 2
S^{-1}\partial_x S \partial_x) \mbox{Im} \chi
] \rangle dx + \nonumber\\
& + &  \int_x^{\infty} S^3 \mathcal{O}(\chi) dx +
4(-\dot X(t)) \int_x^{\infty} \partial_x S S \mbox{Re} \chi dx\nonumber\\
& + & 4 \dot E \int_x^{\infty} S\partial_ES \mbox{Re} \chi dx,
\end{eqnarray}
by the equation $\partial_t \mbox{Re} \chi=H(\mbox{Im} \chi)+{\cal O}(S^2\chi)
$. Integration by parts of the operator
$- \partial_x^2 -2S^{-1}\partial_x S \partial_x$ yields
\begin{equation}\label{der3}
2 \partial_t \int_x^{\infty} S^2 \mbox{Re} \chi dx=
$$$$
2 S^2 \partial_x \mbox{Im} \chi +  \int_x^{\infty} S^2 \langle [
+ \mathcal{O} (\chi S^2)\rangle
]dx
$$$$
+ 4\dot E \int_x^{\infty} S\partial_E S \mbox{Re} \chi dx +
4(-\dot X(t))\int_x^{\infty}S\partial_x S  \mbox{Re} \chi dx.
\end{equation}
The last term can be represented in the form
\begin{equation}\label{der4}
4( -\dot X(t))\int_x^{\infty}S \partial_x S  \mbox{Re} \chi dx
$$$$
= -2\dot
X(t)\int_x^{\infty}\langle \partial_x(S^2 \mbox{Re}
\chi)-S^2(\partial_x \mbox{Re} \chi)\rangle dx
$$$$
= +2\dot X(t)S^2 \mbox{Re} \chi +2\dot X(t)
\int_x^{\infty}S^2(\partial_x \mbox{Re} \chi)dx.
\end{equation}

Starting from the differential equation for $R$ (see \cite{GS05}) we
derive that the quantity
\begin{equation}\label{der5}
\partial_t\int_x^\infty |R|^2=2S^2|\chi|^2\partial_x q + \mathcal{O}(\chi^3
S^3 + \cdots)
\end{equation}
is of the order of $|\chi|^2 S(x)^2\partial_x q$ .

Putting all in Eq. (\ref{mass2}) we get
\begin{eqnarray}\label{der6}
-2S^2 |\chi|^2 \partial_x q & - &\partial_t\int_{-\infty}^{+\infty}|S|^2 dx\  +  \mathcal{O}(\chi^3 S^3 + \cdots) + 2
\dot E \int_{-\infty}^x S \partial_E S dx \nonumber
\\
& = &-2 S^2 \partial_x q  - 4
S^2(\mbox{Re} \chi)
\partial_x q \nonumber
\\
&-& 2S^2 |\chi|^2 \partial_x ( q) + 2 \dot
X(t)\int_x^{\infty} S^2(\partial_x \mbox{Re} \chi) dx + 2 S^2
\partial_x q \nonumber
\\
&+&2S^2|\chi|\cos\eta \partial_x q+2S^2|\chi|\partial_x q\sin^2\eta(2\cos\eta+\sin\eta) \nonumber
\\
& + & 2\int_x^{\infty} S^2 \langle [
 \mathcal{O} (\chi^2 S^m)\rangle ] dx + 4 \dot E \int_x^{\infty} S \partial_E S \mbox{Re} \chi dx.
\end{eqnarray}
\begin{equation}
-\partial_t\int_{-\infty}^{+\infty}|S|^2 dx=-2\int_{-\infty}^{+\infty} S\partial_E S\dot E dx+2\dot X(t)\int_{-\infty}^{+\infty}S\partial_x S dx
$$$$
=-2\dot E\int_{-\infty}^{+\infty} S\partial_E S dx.
\end{equation}
As a result we arrive at the equation
\begin{equation}\label{E dot-final}
\dot E(t)\left[2\int_{-\infty}^xS \partial_ES
dx-\int_{-\infty}^{+\infty} S\partial_E S dx-4\int_x^{\infty}S \partial_E S \mbox{Re} \chi dx \right ]
$$$$
\dot E\left[-2\int_{-\infty}^xS \partial_ES-4\int_x^{\infty}S \partial_E S \mbox{Re} \chi dx \right ]
dx
$$$$
=-2S^2(\mbox{Re} \chi)\partial_x q+2S^2|\chi|\partial_x q\sin^2\eta(2\cos\eta+\sin\eta)
$$$$
-2S^2|\chi|^2\partial_x(\theta+q)+2S^2|\chi|^2\partial_x q
+ 2\dot X(t)\int_x^{\infty}
S^2(\partial_x \mbox{Re} \chi) dx
$$$$
+2\int_x^{\infty} S^2
 \mathcal{O} (\chi^2 S^m)dx.
\end{equation}
for $\dot E(t)$.
We also note that for $\sin\eta=\pm\frac{2}{\sqrt 5}, \cos\eta=\mp\frac{1}{\sqrt 5},$
the $\sin^2\eta$ term is zero to leading order. See Theorem \ref{eta values}.

Next, we find the leading order term on the right hand side of equation (\ref{E dot-final}).
We compute first
$$
\partial_x \mbox{Re}\chi=\cos\eta \partial_x |\chi|-|\chi|\sin \eta \partial_x \eta
$$$$
=\cos\eta \sin\eta v_q (1+{\cal O}(\chi)) +\sin\eta(\sin\eta+\cos\eta+{\cal O}(\chi))v_q
$$$$
=v_q\sin\eta \left[ \cos\eta+(\sin\eta+\cos\eta+{\cal O}(\chi))\right]
$$
Now, at the trapping values of $\eta$, that is, when  $\sin\eta=\pm 2/\sqrt 5,$
the above trigonometric expression vanishes identically. When $\cos\eta=0,$ the
integral multiplying the $\dot X(t)$ factor is seen to be zero, by direct integration by parts, since now
$\mbox{Re} \chi=|\chi|\cos\eta=0,$ by assumption.
Hence
$$
\partial_x \mbox{Re}\chi=v_q {\cal O}(\chi).
$$
Since $\dot X(t)$ is of order $\epsilon\omega\ll 1,$
we conclude that the $\dot X(t)$  term is of higher order correction to the term $-2S^2(\mbox{Re} \chi)\partial_x q.$
All other terms are higher order in $\chi$.

Finally we get our principal result, (at trapped $\cos\eta=\pm\frac{1}{\sqrt 5}$):
\begin{equation}\label{dissipation}
\dot E(t)= \left(\int_{-\infty}^xS\partial_ES\right)^{-1}S^2 \mbox{Re}\chi \partial_x
q + \mathcal{O}(S^2 [\chi^2 v + \dot X (\mbox{Re} \chi)^2 ]).
\end{equation}
At $\cos\eta\sim 0,$ we get that $\mbox{Re} \chi$ is replaced by $|\chi|$, in the above formula.
If $\cos\eta$ is positive at some time, it will stay there, and flow either to zero or to $+\frac{1}{\sqrt 5}.$
The same is true if it is negative, not too far from zero. This is because the equation for the derivative of $\mbox{Re}\chi$ with respect to time, is positive to leading order, for the region where $\partial_x q$ is positive.
Using the soliton stability condition, it follows that $\dot E(t)$
is negative, to the leading order in $\chi$.
\end{proof}

We now need to show that $\chi$ stays small in the relevant
neighborhood of space and time, that is, around the transition
region, including the range of the classical trajectories, that
contribute to the solution of the Euler equation, up to the time
desired. The equation for $\chi$ can be written as
\begin{equation}\label{chieq}
\imath \dot \chi= \left(- \partial_x^2 - 2S^{-1}\partial_x S
\partial_x \right) \chi + 2 S \mbox{Re} \chi + \mathcal{O}(\chi^2)$$$$
 + \mathcal{O}(W)+\mathcal{O}(\dot E\chi)+\mathcal{O}(\dot\gamma \chi).
\end{equation}
This is done in the next section.

\subsection{ Multiplicative Perturbation Bounds}

\paragraph{Bounding $\chi$}.

In this section we derive bounds on $|\chi|$ and $|\chi|'$, using the estimates for $q'$ and the (bootstrap) assumption of small $\chi.$
Recall that
\begin{eqnarray}
1 + |\chi| \cos\eta = |1 + \chi| \cos q \\
|\chi| \sin\eta = |1 + \chi| \sin q
\end{eqnarray}
In particular, it follows that
$$
\sin q=q+{\cal O}(\chi^3)={\cal O}(\chi).
$$
We have the following bounds on $|\chi|'$:

\begin{proposition}
For $0<|\chi|<<1$ ,
$$
|\chi|' \le c|q'|.
$$
$$|\eta'| \le\frac{ c|q'|}{|\chi|}$$
\end{proposition}

\begin{proof}

 For $|\chi| \ll 1$ we may write that (see Section \ref{phases})
\begin{equation}\label{chiprime-b}
|\chi|' = \frac{q'}{|\chi|} |1 + \chi| [+\sin q (1 + \cot\eta) - |\chi|\cos\eta]+\mathcal O(|\chi|q')
\end{equation}
and
\begin{equation}\label{chiprime-c}
|\chi|\sin\eta = |1 + \chi| \sin q.
\end{equation}
 Hence
\begin{equation}\label{chiprime-d}
|\chi|' \leq c|q'| |\sin\eta |
\end{equation}
and,
\begin{equation}\label{chiprime-d1}
|\chi|' =q'\sin\eta+\mathcal O(|\chi|q').
\end{equation}

Now we use
\begin{equation}\label{chiprime-e}
\int_a^b \partial_x |\chi|^2 dx = |\chi(b)|^2 - | \chi(a)|^2 = 2 \int_a^b |\chi| |\chi|' dx
\end{equation}
and from equation (\ref{eta'}) of Section 7
$$
\eta' = -\frac{q'}{|\chi|} \left[\cos\eta + \sin\eta +\mathcal O(|\chi|)\right] \leq c \frac{q'}{|\chi|}
$$
for $|\chi| \geq \varepsilon |q'|$.
To see that the higher order terms do not blow up at the limits when either $\sin\eta$ or $\cos\eta$ approach zero,
we use the asymptotic bound on $|\chi|'$, in equations (\ref{phase9a})  and (\ref{phase9b})  of Section 7:
In the limit when $\cos\eta$ approaches zero, we get that
$$
\eta'|\chi|^2\sin\eta\thickapprox q'-|\chi|'+\mathcal O(q'|\chi|)=\mathcal O(q'|\chi|).
$$
Hence in this limit we get
$$|\eta'| \leq c|q'|/|\chi|$$.
Similarly, in the limit $\sin\eta$ near zero, we get:
$$
\eta'|\chi| \thickapprox q'+|\chi|'\thickapprox cq'.
$$

 Hence ($q|\chi|\le 2$ by (\ref{chiprime-c}) and $|\chi|\ll 1$),
$$
\partial_x \chi = \partial_x \left(|\chi| e^{i\eta}\right) = e^{i\eta}|\chi|' + i\eta'e^{i\eta}|\chi| = {\cal O}(q')
$$

\end{proof}
\begin{theorem}\label {eta values}
For $0<|\chi|<<1$ ,
we have that
$
|\cos\eta|\thickapprox 0 $   \textbf{or} $ \frac{1}{\sqrt 5}$ \textbf{and }$\sin\eta=-2\cos\eta$.

\end{theorem}
\begin{proof}
We have that
$$
\partial_x q(1+{\cal O}(\chi))=\partial_x \mbox{Im}\chi=\partial_x (|\chi|\sin\eta)=|\chi|'\sin\eta+|\chi|\cos\eta\partial_x \eta
$$$$
=\sin^2\eta\partial_x q(1+{\cal O}(\chi))-\cos\eta(\sin\eta+\cos\eta+{\cal O}(\chi))\partial_x q
$$$$
=\left[\sin^2\eta-\cos\eta(\sin\eta+\cos\eta)+{\cal O}(\chi)\right]\partial_x q.
$$
It follows that
$$
\sin^2\eta-\cos\eta(\sin\eta+\cos\eta)+{\cal O}(\chi)=1.
$$
The solutions of this trigonometric identity (up to ${\cal O}(\chi)$),
are $\cos\eta=0,$ and $\tan\eta=-2$,
from which the result of the theorem follows.
\end{proof}
\textbf{Remark}
As long as $\chi$ is small, the solution is trapped around one of the above solutions of the trigonometric identity.
The case of $\cos\eta=0$ corresponds to essentially no mass change of the soliton energy; the other cases corresponds to the soliton either gaining or losing, monotonically, energy (and mass)

We can now prove the main bound on $\chi$:

\begin{theorem}{Multiplicative Error Bound}

For $t\le \frac{\delta}{\omega\varepsilon},$ and $x$ is in the transition region,
$$
|\chi(x)|\le \delta.$$
\end{theorem}

\begin{proof}

Equation (\ref{chieq}) yields
$$
i \partial_t \langle F(x), \chi \rangle =
$$$$
\langle F(x), - \partial_x^2\chi \rangle + \langle F, W(x,t) \chi\rangle - 2 \langle F(x), \frac{\partial_x S}{S} \partial_x \chi\rangle + 2 \langle F, S \mbox{Re}\chi\rangle + h.o.t.
$$
Therefore, for $F$ supported in $[a,b]$ we have
$$
|\langle F(x), - \partial^2_x\chi\rangle| = |\langle \partial_xF, \partial_x \chi\rangle| = \left| \int_a^b \partial_x F(x) \partial_x\chi dx\right| \leq \sup_{x \in[a,b]} |\partial_x \chi| \int_a^b |\partial_x F(x)| dx
$$$$
\leq C|q'| (b - a) \sup_{x \in[a,b]} |\partial_x F(x)| \leq C |q'|
$$
by choosing
$$
F(x) =\left\{
\begin{array}{cc}
\displaystyle \alpha x, & 0 \leq x \leq \frac{b - a}{4};\\
&\\
\displaystyle1, & \frac{b-a}{4} \leq x \leq \frac{3(b - a)}{4};\\
&\\
\displaystyle -\alpha x, & \frac{3(b-a)}{4} \leq x \leq b
\end{array}
\right.
$$
so that
$$
\int_a^b |\partial_x F(x)| dx \leq C
$$
does not depend on $b - a$.

The second term is higher order in the sense that
$$
\langle F, W(x,t) \chi\rangle = {\cal O}\left(e^{- 1/\omega}\right) ||F\chi||.
$$
The third term is bounded similarly
$$
\left|\langle F(x), \frac{\partial_x S}{S} \partial_x \chi \rangle\right| \leq C \int_a^b |\partial_x \chi| \leq C|b -a|q'.
$$
The rest of the terms are higher order or ${\cal O}\left(e^{- 1/\omega}\chi\right)$

These estimates together with equation (\ref{chieq}) imply that
\begin{equation}\label{chiprime-f}
|\partial_t\langle F,\chi \rangle| \leq C |b - a| q' + C_1 q'.
\end{equation}
So,
\begin{equation}\label{chiprime-g}
|\langle F, \chi\rangle_t| \leq C |b - a| \int_0^t q'(s) ds + C_1 \int_0^t q'(s) ds
\end{equation}
for $t \leq 2 t_0$, and recall that
$$
v = q'(t) \sim \varepsilon^2 \omega^2 t, \ \mbox{for}\ \ t\leq t_0.
$$
More generally for $t \leq t_0$ we have
$$
\int_0^t q'(s) ds \sim \varepsilon^2 \omega^2 t^2/2, \ \mbox{for}\ t < t_0.
$$

On an interval of size one we use the condition $\partial_x \chi = {\cal O}(q') \ll 1$ in order to estimate
$$
\int_a^{a + 1}|\chi| dx = \int_a^{a + 1}e^{-i\eta(x)}\chi dx =
$$$$
\left. e^{-i\eta(x)} \int_a^x \chi(y) dy\right|_a^{a + 1} + \int_a^{a+1} i \eta'(y) \int_a^y \chi(y') dy'.
$$
Then
\begin{equation}\label{chiprime-h}
\left|\int_a^{a + 1}|\chi| dx\right| \leq \left| \int_a^{a +1} \chi(y) dy\right| + \int_a^{a + 1} |\eta'(y)| dy \sup_{a \leq y \leq a+1} \left|\int_a^y \chi(y') dy'\right|.
\end{equation}

Now, if at some point $|\chi(a)| \leq C q'$, then it follows from (\ref{chiprime-e}) and (\ref{chiprime-j}) below that
\begin{equation}\label{chiprime-i}
\sup_{[a,b]} |\chi| \leq C |b - a| q' + C_1 q'.
\end{equation}
Therefore we only need to consider an interval $[a,b]$ with
$$
|\chi(y)| > C q' \Rightarrow |\eta'(y)| \leq \frac{2q'}{|\chi|} \leq C < \infty.
$$
We have used that equation (\ref{chiprime-e}) implies
$$
\frac{|\chi(b)|^2 - |\chi(a)|^2}{\displaystyle\sup_{x \in [a,b]}|\chi|} \leq C |b - a| q'.
$$
We make $|\chi(b)| = \sup_{x \in [a,b]} |\chi|$ by varying $b$, and as a result we get
\begin{equation}\label{chiprime-j}
\sup_{x \in [a,b]} |\chi| - \frac{|\chi(a)|^2}{\displaystyle \sup_{x \in [a,b]} |\chi|} \leq C |b - a| q'.
\end{equation}

Then by \ref{chiprime-f}
\begin{equation}\label{chiprime-k}
\int_a^{a+1}|\chi| dx \leq \sup_{y \in [a,a+1]} \left|\int_a^y \chi(y') dy'\right| \leq C |b - a| \int_0^t q'(s) ds + C_1 \int_0^t q'(s) ds.
\end{equation}
So, in particular,
\begin{equation}\label{chiprime-m}
|\chi(y)| \leq \max \left\{C q' + C _1 \int_0^t q'(s) ds\right\} \leq C \varepsilon^2\omega^2t^2, \ \mbox{for}\ t\leq t_0 = \frac{1}{\omega\varepsilon}
\end{equation}
which means that $|\chi(x)| \leq \delta $ for all $t \leq \delta/(\omega\varepsilon)$.
\end{proof}

\section{Momentum decay}

In this section we prove that if the radiation is outgoing from the boundary of the well, and if the soliton energy is monotonic decreasing
($\dot E\le 0$), then the soliton slows down, monotonically. To this end, we modify the Ansatz so as to include the slow time dependent change in the oscillation, explicitly in the soliton term. That is we replace $S(x-X)$ by $S(x-a(t)),$
with $a(t)\equiv \int_0^t v(s) ds+ D(t).$
\begin{theorem}
Assume that for the NLS as before, we have that $\partial_x q >0$, that $\chi$ is small, and that $\dot E<0.$
Then the velocity of the center of the soliton, $v(t)$, is monotonic decreasing in time, according to the {\emph Dissipation equation} below.
\end{theorem}
\begin{proof}
Let
$$
i \frac{\partial \psi}{\partial t} = (- \partial_x^2 + V) \psi + \lambda
|\psi|^2\psi
$$
$$
-ES_E = (- \partial_x^2 + V) S_E + \lambda S_E^3
$$
where $\lambda<0$. The function $S$ of the form
$$
S_E(x,t) = e^{i\theta_h(x,t)} S_E(x - a(t))
$$
solves NLS equation if ($V=V_h)$
$$
\theta_h = \int_0^t E(s) ds + \gamma(t) + v(t)x/2.
$$

{\em Momentum identity}

 Let $F_K = F_K(|x| \geq K)$ and such that $|F_K^{(n)}|\lesssim K^{-n}$.
  Then
$$
\partial_t \langle \psi, F_K p F_K \psi \rangle = - \langle \psi,
F^2_K \frac{\partial V}{\partial x} \psi \rangle +  \langle
\psi,\left [ F_K \lambda (- \partial_x |\psi|^2) F_K+\imath[- \partial_x^2 ,F_KpF_K]\right ] \psi \rangle
$$$$
= - \langle \psi, F_K^2 \frac{\partial V}{\partial x} \psi \rangle +\langle
\psi,\imath[- \partial_x^2 ,F_KpF_K] \psi \rangle
+\frac{\lambda}{2} \int \partial_x F_K F_K |\psi|^4 dx
$$
where
$$
p = -i\partial_x.
$$

 Ansatz:
$$
\psi = e^{i\theta_h(x,t)} \left[S_{E(t)}(x - a(t)) + R(x,t)\right]
$$
with orthogonality as before.

Then

\begin{equation}\label{R}
i\frac{\partial R}{\partial t} = (- \partial_x^2 + V_t - E(t)) R + 2\lambda S^2 R +
\lambda S^2 \overline{R}
$$$$
 + \lambda |R|^2 R + {\cal O}(SR^2) + {\cal O}(\dot
E, \dot\gamma)R +(V-V_t+x\dot v/2)(R+S).
\end{equation}

$$ V_t=V(x-a(t)).$$
Now we use the Ansatz in the momentum identity
\begin{equation}\label{S1}
\partial_t \langle e^{i\theta}S, F_K p F_K e^{i\theta}S \rangle =\partial_t \langle S, F_K (p+v/2) F_K S \rangle=
\partial_t \langle S, F_K (v/2)F_K S \rangle+0
$$$$
= - \partial_t \langle e^{i\theta} R, F_K p F_K e^{i\theta} R
\rangle - \langle S, F^2_K \frac{\partial V}{\partial x} S \rangle -
\langle R, F^2_K \frac{\partial V}{\partial x} R \rangle +
\frac{\lambda}{2} \int F_K' |\psi|^4 dx
$$$$
- \left[\partial_t \langle e^{i \theta} S, F_K p F_K e^{i \theta}R\rangle +
c.c.\right]-2 \mbox{Re} \langle S,F_K^2\frac{\partial V}{\partial x} R \rangle+\langle\psi, \imath [-\partial_x^2 ,M_K]\psi\rangle.
\end{equation}

\begin{equation}\label{Com}
\langle\psi, \imath [- \partial_x^2,M_K]\psi\rangle=\langle e^{\imath\theta} S,\imath [- \partial_x^2,M_K] e^{\imath\theta} S \rangle
$$$$
+ \langle e^{\imath\theta} R,\imath [- \partial_x^2,M_K] e^{\imath\theta} R \rangle+ 2 \mbox{Re} \langle e^{\imath\theta} S,\imath [-\partial_x^2,M_K] e^{\imath\theta} R \rangle.
\end{equation}

\textbf{Estimates}:

Define $$
M_K\equiv F_K pF_K.
$$
\begin{equation}\label{S2}
\partial_t \langle S, F_K^2 (v/2) S\rangle- \langle e^{\imath \theta} S,\imath[-\partial_x^2 +V+\lambda S^2,M_K]e^{\imath\theta} S\rangle
$$$$
=(\dot v/2) \langle S,F_K^2 S\rangle +(v/2)\partial_t\langle S,F_K^2 S\rangle-\langle S,\imath [(p+v/2)^2+V+\lambda S^2,F_K(p+v/2)F_K] S \rangle
$$$$
=(\dot v/2) \langle S,F_K^2 S\rangle + T1-\langle S,\imath [V-V_t+\omega^2 a x, M_K] S\rangle+ \langle S, F_K^2 S\rangle \omega^2 a(t)
$$$$
-\langle S,\imath[p^2+V_t+\lambda S^2,F_K(p+v/2)F_K]S\rangle-\langle S,\imath [pv,F_K(p+v/2)F_K] S\rangle
$$$$
=(\dot v/2) \langle S,F_K^2 S\rangle + T1- v\langle S,[F'_KpF_K+F_KpF'_K] S\rangle+ \langle S, F_K^2 S\rangle \omega^2 a(t)
$$$$
+{\cal O}(S^4) + \langle S,s_n(x) V_{(3)}^2 S\rangle - v^2\langle S, \imath [p,F_K^2] S\rangle =\partial_t [ v/2 \langle S,F_K^2 S\rangle] +
$$$$
{\cal O}(S^4)+{\cal O}(S^2) {\cal O}(v^2+a)+\omega^2 a(t) \langle S,F_K^2 S\rangle+ \langle S,s_n(x) V_{3}^2 S\rangle.
\end{equation}
Here $s_n(x) $
is the sign of $x$.
We used that
$$
F_K \imath [p, V-V_t+\omega^2 a]F_K\equiv s_n(x) V_{(3)}^2F_K^2 .
$$
Here, we applied the Virial Theorem for the commutator expectation on $S$. The $v^2$ term drops since it is the integral of a product of the symmetric function $S^2$, (up to corrections of order $a$) with the antisymmetric function $(F_K^2)'$, the last dropped since the real part of the expectation of $p$ , on real functions , vanishes.
\begin{equation}\label{R1}
- \partial_t \langle e^{i\theta} R, F_K p F_K e^{i\theta} R \rangle+\langle e^{\imath\theta} R,\imath [- \partial_x^2 +V,M_K] e^{\imath\theta} R\rangle
$$$$
= - \partial_t \langle R, (v(t)/2) F_K^2 R \rangle - \partial_t \langle R,
F_K p F_K R \rangle+\langle e^{\imath\theta} R,\imath [- \partial_x^2 +V,M_K] e^{\imath\theta} R\rangle
$$$$
=-(\dot v/2)\langle R,F_K^2 R\rangle -(v/2)\partial_t \langle R,F_K^2 R\rangle-\partial_t\langle R,F_K p F_K R\rangle
$$$$
+\langle R,\imath[(p+v/2)^2+V,F_K(p+v/2)F_K]
R\rangle
$$$$
=-(\dot v/2)\langle R,F_K^2 R\rangle -(v/2)\partial_t \langle R,F_K^2 R\rangle -\langle R,\imath[-\partial_x^2 + V_t-E+2\lambda S^2,M_K] R \rangle
$$$$
-\langle R, \imath [V-V_t+xv/2,M_K] R\rangle-\langle (V-V_t+xv/2)S,M_K R\rangle-\langle R,M_K (V-V_t+xv/2)S\rangle
$$$$
+\langle R \left[ \imath[p^2+V,M_K+(v/2)F_K^2]+\imath [pv,F_K(p+v/2)F_K]+{\cal O}(SR+S^2C,\dot E,\dot \gamma)M_K \right] R\rangle
$$$$
= -(v/2)\partial_t \langle R,F_K^2 R\rangle+v\langle R,\left[ \imath[p^2,(F_K^2/2)]
+(\partial_x F_K(p+v/2)F_K+F_K(p+v/2)\partial_x F_K)\right] R \rangle
$$$$
+\langle R, {\cal O}(S^2,SR,O(\dot E,\dot \gamma)) M_K R\rangle-\langle (V-V_t+xv/2)S,M_K R\rangle-\langle R,M_K (V-V_t+xv/2)S\rangle
$$$$
=-(v/2)\partial_t \langle R,F_K^2 R\rangle+v\langle R,\left[ \imath[p^2,(F_K^2/2)]+(\partial_x F_K p F_K + F_K p \partial_xF_K)\right] R \rangle
$$$$
+v^2\langle R,\partial_x F_KF_K R\rangle +h.o.t.-\langle (V-V_t+xv/2)S,M_K R\rangle-\langle R,M_K (V-V_t+xv/2)S\rangle
$$$$
=-(v/2)\partial_t \langle R,F_K^2 R\rangle+(v^2/2)\langle R,\partial_x F_KF_K R\rangle +(v^2/4)\langle R,F_K\partial_x F_K R\rangle+h.o.t
$$$$
+\imath \langle R,[p^2,F_K^2] R\rangle v -2 \mbox{Re} \langle \tilde V S,pR\rangle.
\end{equation}
Here $ \tilde V\equiv (V-V_t+xv/2)$.
We also chose $F_K \tilde V=\tilde V, \partial_x F_K \tilde V=0.$
We have used in the above the following:
$$
\partial_x F_KpF_K+F_Kp \partial_xF_K = \imath[p,F_KpF_K]=\imath [p^2,F_K^2]/2
$$

%


%

\textbf{T1}

\begin{equation}\label{S3}
(v(t)/2) \partial_t \langle S, F_K^2 S\rangle = v(t)\dot E \int
\frac{\partial S}{\partial E} F_K^2 S dx - v(t)\dot a \int \frac{\partial
S}{\partial x} S F_K^2 dx
$$$$
= v(t)\dot E \int \frac{\partial S}{\partial E} F_K^2 S dx +
 v(t)\dot a \int F_K S^2 \partial_x F_K dx
$$$$
= v(t)\dot E \int \frac{\partial S}{\partial E} F_K^2 S dx +
( v(t)/2)\dot a \int F_K [- S(x)^2 + S(x-a)^2] \partial_x F_K dx
$$$$
 +(v(t)/2)\dot a \int F_K S(x)^2 \partial_x F_K dx (=0)
$$$$
= v(t)\dot E \int \frac{\partial S}{\partial E} F_K^2 S dx - v(t)\dot a
a(1 + {\emph o}(a)) \int F_K S \frac{\partial S(x)}{\partial x} \partial_x F_K dx
\end{equation}

Finally, we derive the statement of the theorem and the  \textsl{Dissipation Equation}:
\medskip

\textbf{Dissipation Equation}

\begin{equation}\label{Diss}
\frac{d}{dt}\left[ (v/2)\langle S,F_K^2 S\rangle \right]
$$$$
=-(v/2)\partial_t\langle R,F_K^2 R\rangle -\omega^2 a(t) \langle S,F_K^2 S\rangle +v^2\left[\langle R,F_K\partial_x F_K R\rangle \right]
$$$$
+ \langle \psi,s_n(x)V_{(3)}^2 \psi\rangle-\langle R,s_n(x)V_{(3)}^2 R\rangle+ h.o. t.
$$$$
\end{equation}

We used that due to antisymmetry of $\partial_x F_K,$ $\langle S, F_K \partial_xF_K S\rangle={\cal O}(a).$

\textbf{Cross Terms}

\begin{equation}\label{Cross1}
2 \mbox{Re} \langle e^{\imath\theta}S, \imath [- \partial_x^2 ,M_K]e^{\imath\theta}R\rangle-\partial_t\left[ \langle e^{\imath\theta}S,  M_K e^{\imath\theta}R\rangle
+c.c.\right]
$$$$
=2 \mbox{Re} \langle S, \imath [p^2+pv,M_K] R\rangle-\partial_t\left[\langle S, F_K(p+v/2)F_K R\rangle +c.c.\right]
$$$$
=2 \mbox{Re} \langle S, \imath [p^2+pv,M_K] R\rangle-\left[\langle \partial_ES \dot E- \partial_x S \dot a,F_K(p+v/2)F_K R\rangle +c.c.\right]
$$$$
+\left[\dot v\langle S,F_K^2 R\rangle +c.c.\right] +\langle S,F_K(\partial_x)F_K(-2\partial_t \mbox{Im} R)\rangle -(v/2)\langle S,F_K^2 2 \mbox{Re} \dot R\rangle
$$$$
=\left[\dot v\langle S,F_K^2 R\rangle +c.c.\right]+\dot a\langle \partial_x S,F_K^2 \mbox{Re} R\rangle (v/2)-\dot E v\langle \partial_E S,F_K^2 \mbox{Re} R\rangle
$$$$
+\mbox{Re} \left [ \dot a \langle \partial_x S, M_K  R\rangle-\dot E\langle \partial_E S,M_K  R\rangle\right].
\end{equation}

We used above the following observations:

\textbf{Term 2}

For $f,g$ real valued,
$$
 \langle S, F_KpF_K S\rangle = 0, \mbox{Re} \langle f,pg\rangle=0,
$$
since $S F_K$ is real function.

\textbf{Term 3}

$$
\langle S, F_K^2 \frac{\partial V}{\partial x} S\rangle = \int S^2
\frac{\partial V}{\partial x} F_K^2 dx = \int [S^2(x-a) - S^2(x)]
\frac{\partial V}{\partial x} F_K^2 dx
$$$$
= -2 a(1 + {\emph o}(a)) \int S(x) \frac{\partial S(x)}{\partial x}
\frac{\partial V}{\partial x} F_K^2 dx
$$
where we have used
$$
\int S^2(x) \frac{\partial V}{\partial x} F_K dx = 0
$$
if $V(x) = V(-x).$

\textbf{ Term 4}

\begin{equation}
\mbox{Re} \langle \tilde V S,pR\rangle= \langle \tilde V S,\partial_x \mbox{Im} R\rangle\sim\langle S^2 \tilde V'/2, \mbox{Im} \chi\rangle
$$$$
={\cal O}\left (\partial_x S(K)S(K)\omega^2(|\mbox{Im} \chi|+|v_q|)\right )a(t)={\cal O}\left (\partial_x S(K)S(K)\omega^2(\epsilon^2\omega^2t+\delta)\right ),
\end{equation}

for $ |t|\le \frac{\delta}{\epsilon\omega}.$

\medskip

\textbf{Term 5 (cross term)}

\begin{equation}
\langle S,F_K \partial_x F_K \partial_t S \mbox{Im} \chi\rangle = \langle F_K S,\partial_t F_K S \partial_x \mbox{Im} \chi \rangle+ \langle F_K S,\partial_t\partial_x(F_K S)\mbox{Im} \chi\rangle
$$$$
=\langle \partial_t(F_KS)^2/2,v_q\rangle+ \langle (F_KS)^2,\dot v_q\rangle +
$$$$
\langle (F_KS)(\partial_t\partial_x (F_KS) \mbox{Im} \chi) \rangle- \langle(F_KS)^2_x/2,\partial_t\partial_x (\mbox{Im} \chi)\rangle
$$$$
+ \langle(F_KS)^2_x/2,\mbox{Im} \chi\rangle = {\cal O}\left ( \epsilon^2\omega^2 S\partial_x S(K)[(\dot a+ \dot E) t+1]\right)a(t).
\end{equation}

\textbf{term 6 (cross term)}

\begin{equation}
v\langle S,F_K^2 \mbox{Re} \dot R \rangle= v\langle S,F_K^2 H \mbox{Im} \chi S\rangle - v \langle S,F_K^2\tilde V S\rangle
$$$$
v\langle S,[F_K^2,- \partial_x^2]S \mbox{Im} \chi\rangle - v\langle S,F_K^2\tilde VS\rangle
$$$$
= {\cal O}(\frac{1}{K})vS^2(K)v_q+ {\cal O}(v)aS^2(K)={\emph o}(1)v S^2(K),
\end{equation}
using that
$$
\int \tilde V=0
$$
on support of the function $F_K$, and $F_K^2S^2$ is symmetric up to corrections of order $a$.

\end{proof}

\section{Calculation of phases}
\label{phases}
We analyze here equations (\ref{phase5a}) and (\ref{phase5b}):
\begin{eqnarray}
1 + |\chi| \cos\eta = |1 + \chi| \cos q \\
|\chi| \sin\eta = |1 + \chi| \sin q
\end{eqnarray}

 We need the derivative
$$
\partial_x|1 + \chi| = \frac{1}{|1 + \chi|}\left[|\chi||\chi|' + \partial_x \mbox{Re}\chi \right]
$$
$$
= \frac{1}{|1 + \chi|}\left[|\chi||\chi|'+|\chi|'\cos\eta-|\chi|\sin\eta'\right],
$$
and
$$
1 - \frac{1}{|1 + \chi|} = \frac{|1 + \chi| - 1}{|1 + \chi|} = \frac{|1 + \chi|^2 - 1}{|1 + \chi| [1 + |1 + \chi|]}
= \frac{2\mbox{Re} \chi + |\chi|^2}{|1 + \chi| [1 + |1 + \chi|]}
$$$$
=  [\mbox{Re} \chi + |\chi|^2/2] [1-3 \mbox{Re} \chi + O{|\chi|^2}] = |\chi| \cos \eta+ |\chi|^2/2-3|\chi|^2\cos^2\eta  + O(|\chi|^3)
$$

Now differentiate equations (\ref{phase5a}) and (\ref{phase5b}) with respect to $x$ and get
\begin{eqnarray}
|\chi|' \cos\eta - \eta' |\chi| \sin\eta = |1 + \chi|' \cos q - |1 + \chi| \sin q q' \label{phase9a}\\
|\chi|' \sin\eta + \eta' |\chi| \cos\eta = |1 + \chi|' \sin q + |1 + \chi| \cos q q' \label{phase9b}
\end{eqnarray}

Then
\begin{equation}
|\chi|' \cos\eta \left[1 - \frac{\cos q+|\chi|/\cos \eta}{|1 + \chi|}\right] - \eta' |\chi| \sin\eta \left[1 - \frac{\cos q}{|1 + \chi|}\right]
$$$$
= - |1 + \chi|q' \sin q  
|\chi|' \left[\sin \eta - (\cos\eta+|\chi|) \frac{\cos q}{|1 + \chi|}\right] +
$$$$
\eta' |\chi| \left[\cos\eta + \sin\eta \frac{\sin q}{|1 + \chi|} \right] = |1 + \chi|q' \cos q  \label{phase10b}
\end{equation}

Now we calculate the determinant of this equation
\begin{equation}
Det \left|
\begin{array}{cc}
\cos\eta \left[1 - \frac{\cos q+|\chi|/\cos\eta}{|1 + \chi|}\right]; & - |\chi| \sin\eta \left[1 - \frac{\cos q}{|1 + \chi|}\right] \\
\left[\sin \eta - (\cos\eta+|\chi|) \frac{\cos q}{|1 + \chi|}\right]; & |\chi| \left[\cos\eta + \sin\eta \frac{\cos q}{|1 + \chi|}\right]
\end{array}
\right| =
$$$$
-(\sin\eta+\cos\eta)\frac{|\chi|^2\cos q}{|1+\chi|}
+|\chi| \left[1 - \frac{\cos q}{|1 + \chi|}\right] \cdot
$$$$
 \left[\cos\eta \left[\cos\eta +
\sin\eta \frac{\cos q}{|1 + \chi|}\right] + \sin\eta \left[\sin \eta - \cos\eta \frac{\cos q}{|1 + \chi|}\right] \right]
$$$$
 =
-(\sin\eta+\cos\eta)\frac{|\chi|^2\cos q}{|1+\chi|}+
|\chi| \left[1 - \frac{\cos q}{|1 + \chi|}\right]
$$$$
= -|\chi|^2 \sin\eta \left[1 + O\left(\frac{q^2}{|\chi|}\right)\right]-|\chi|^3\cos^2\eta.
\end{equation}

In order to find $|\chi|'$ we need the determinant
\begin{equation}
Det \left|
\begin{array}{cc}
- |1 + \chi| \sin q q'; & - |\chi| \sin\eta \left[1 - \frac{\cos q}{|1 + \chi|}\right] \\
|1 + \chi| \cos q q' ; & |\chi| \left[\cos\eta + \sin\eta \frac{\cos q}{|1 + \chi|}\right]
\end{array}
\right| =
$$$$
q' |\chi| |1 + \chi| \left[- \sin q \left[\cos\eta + \sin\eta \frac{\cos q}{|1 + \chi|}\right] + \sin\eta \left[1 - \frac{\cos q}{|1 + \chi|}\right] \cos q \right] =
$$$$
q' |\chi| |1 + \chi| \left[-\sin q (\cos\eta + \sin\eta) + |\chi|\sin\eta \cos\eta \right]
$$$$
+\mathcal O(q'|\chi|^3\sin\eta(\sin\eta+\cos\eta)).
\end{equation}
Then
$$
|\chi|' = \frac{q'} {|\chi|} |1 + \chi| \left[\sin q (1 + \cot\eta) - |\chi|\cos\eta \right]
$$

In order to find $\eta'$ we need the determinant
\begin{equation}
Det \left|
\begin{array}{cc}
\cos\eta \left[1 - \frac{\cos q+|\chi|/\cos\eta}{|1 + \chi|}\right]; & - |1 + \chi| \sin q q' \\
\left[\sin \eta - (\cos\eta+|\chi|) \frac{\cos q}{|1 + \chi|}\right]; & |1 + \chi| \cos q q'
\end{array}
\right| =
$$$$
-q'|\chi|\cos q(\sin q+\cos q)+
$$$$
q' |1 + \chi| \left[\cos q \cos\eta \left[1 - \frac{\cos q}{|1 + \chi|} \right] + \sin q \left[\sin \eta - \cos\eta \frac{\cos q}{|1 + \chi|}\right] \right] \approx
$$$$
-q'|\chi|+q' |1 + \chi| \left[|\chi| \cos^2\eta + q \left[\sin \eta - \cos\eta\right] \right]\approx
$$$$
-q'|1+\chi||\chi|\sin\eta[\sin\eta+\cos\eta](1+\mathcal O(|\chi|)).
\end{equation}

Then
\begin{equation}\label{eta'}
\eta' = -\frac{q'}{|\chi|}  \left[ \cos\eta + \sin\eta+\mathcal O(|\chi|) \right]
\end{equation}


\begin{thebibliography}{99}

\bibitem{GS05} Gang Zhou, Sigal I. M., Reviews in Mathematical Physics, {\bf 17}, 1143, (2005).

\bibitem{GS06} Gang Zhou, Sigal I. M., Geometric and Functional Analysis, {\bf 16}, 1377 (2006).

\bibitem{GS07} Gang Zhou, Sigal I. M., Advances in Mathematics, {\bf 216}, 443 (2007).

\bibitem{GW} Gang Zhou, M. I. Weinstein, Appl. Math. Res. Express {\bf 2011}, 123-181 (2011).

\bibitem{SWGA} A. Soffer, M.I. Weinstein, Geometric and Functional Analysis GAFA, {\bf 8}, 1086 (1998).

\bibitem{SW99} A. Soffer, M.I. Weinstein, Invent. Math. {\bf 136}, 9 (1999).

\bibitem{FGJS} J. Fr\"ohlich, S. Gustafson, B. L. G. Jonsson, I. M. Sigal, Comm. Math. Phys. {\bf 250}, 613 (2004).

\bibitem{FS05} V. Fleurov and A. Soffer, Europhys. Letts. {\bf 72}, 287 (2005).

\bibitem{DFSS07}    G. Dekel, V. Fleurov, A. Soffer, C. Stucchio, Physical Review A {\bf 75}, 043617 (2007).

\bibitem{DFSF09} G. Dekel, O. V. Farberovich, A. Soffer, V. Fleurov, Physica D: Nonlinear Phenomena {\bf 238}, 1475 (2009).

\bibitem{DFFS10} G. Dekel, V. Farberovich, V. Fleurov, A. Soffer, Phys. Rev. A {\bf 81}, 063638 (2010).

\bibitem{r05} O. S. Rozanova, Proc. Amer. Math. Soc. {\bf 133}, 2347 (2005).

\bibitem{K09} Usama Al Khawaja, Physics Letters A {\bf 373}, 2710 (2009).

\bibitem{BS03} V. Buslaev, C. Sulem, Annales de l'Institut Henri Poincar\'e, {\bf 20}, 419 (2003).

\bibitem{SW04} A. Soffer, M.I. Weinstein, Rev. Math. Phys., {\bf 16}, 977 (2004).

\bibitem{m27} E. Madelung,  Z. Phys. {\bf 40}, 322 (1927)

\bibitem{m75} J. H. Marburger, Progr. Quant. Electr., {\bf 4}, 35
(1975).

\bibitem{Sof} A. Soffer, Communications in Partial Differential Equations, {\bf 33}, 1953 (2008).

\bibitem {GJ} I. Gamba, A. J\"ungel, Archive for Rational Mechanics and Analysis {\bf 156}, 183 (2001).

\bibitem{Gols} F. Golse, P.L. Lions, B. Perthame, and R. Sentis. J. Funct. Anal., {\bf 88}, 110 (1988).

\bibitem {Per} B. Perthame, Bull. Amer. Math. Soc. {\bf 41}, 205 (2004).

\bibitem{BGP} D. Bambusi, S. Graffi, T. Paul,  Asymptotic Analysis, {\bf 21}, 149 (1999).

\bibitem {DLev} C. Doering, J. Gibon, C. D. Levermore, Physica D, {\bf 71}, 285 (1994).

\end{thebibliography}
\end{document}